\documentclass[sigconf,authorversion]{acmart}

\usepackage{booktabs} % For formal tables

\setcopyright{rightsretained}
\acmConference[CIKM'19]{CIKM conference}{November 2019}{Beijing, China}
\acmYear{2019}

\usepackage{subfig}
\usepackage{algorithm}
\usepackage[noend]{algpseudocode}
\newtheorem{remark}{\bf Remark}
\newtheorem{assumption}{\bf Assumption}
\newtheorem{lemma}{\bf Lemma}
\newtheorem{prop}{\bf Proposition}
\newtheorem{model}{\bf Model}

\newcommand{\used}{R_u}
\newcommand{\outcome}{Y_u}
\newcommand{\effect}{\tau_u}
\newcommand{\control}{c_u}

\newcommand{\istreated}[1]{Z_{#1}}
\newcommand{\whole}{k}

\newcommand{\userN}{N}

\newcommand{\E}{\mathbb E}

\renewcommand{\t}{T}
\renewcommand{\c}{C}
\usepackage{scalerel,stackengine}
\stackMath
\newcommand{\reallywidehat}{\widehat}
\renewcommand{\bar}{\overline}
\renewcommand{\hat}{\widehat}
\newcommand{\ATE}{\Delta}

\newcommand{\double}{\textrm{double}}
\newcommand{\scaled}{\textrm{scaled}}
\newcommand{\jack}{\textrm{jack}}
\newcommand{\trajectory}[1]{}
\begin{document}
\title{On Heavy-user Bias in A/B Testing}

\author{Yu Wang}
\orcid{0000-0002-5329-7739}
\affiliation{%
  \institution{UC Berkeley}
  \city{Berkeley}
  \state{California}
  \country{U.S.}
}
\email{wang.yu@berkeley.edu}

\author{Somit Gupta}
\affiliation{%
  \institution{Microsoft}
  \city{Redmond}
  \state{Washington}
  \country{U.S.}
}
\email{Somit.Gupta@microsoft.com}

\author{Jiannan Lu}
\orcid{0000-0002-8839-6024}
\affiliation{%
  \institution{Microsoft}
  \city{Redmond}
   \state{Washington}
  \country{U.S.}}
\email{jiannl@microsoft.com}

\author{Ali Mahmoudzadeh}
\affiliation{%
  \institution{Microsoft}
  \city{Redmond}
   \state{Washington}
  \country{U.S.}
}
\email{Ali.Mahmoudzadeh@microsoft.com}
\author{Sophia Liu}
\affiliation{%
 \institution{Microsoft}
 \city{Redmond}
 \state{Washington}
 \country{U.S.}}
\email{Sophia.Liu@microsoft.com}
\renewcommand{\shortauthors}{Y. Wang et al.}

\begin{abstract}
On-line experimentation (also known as A/B testing) has become an integral part of software development. To timely incorporate user feedback and continuously improve products, many software companies have adopted the culture of agile deployment, requiring online experiments to be conducted and concluded on limited sets of users for a short period. While conceptually efficient, the result observed during the experiment duration can deviate from what is seen after the feature deployment, which makes the A/B test result biased. In this paper, we provide theoretical analysis to show that heavy-users can contribute significantly to the bias, and propose a re-sampling estimator for bias adjustment. 
\end{abstract}

%\begin{CCSXML}
%<ccs2012>
%<concept>
%<concept_id>10002950.10003648.10003688.10003693</concept_id>
%<concept_desc>Mathematics of computing~Time series analysis</concept_desc>
%<concept_significance>300</concept_significance>
%</concept>
%</ccs2012>
%\end{CCSXML}

%\ccsdesc[300]{Mathematics of computing~Time series analysis}
\maketitle

\section{Introduction}
\label{sec:intro}

A/B tests or online controlled experiments are used  by a large number of software and technology companies \cite{hohnhold2015focusing,kharitonov2017learning} 
to evaluate changes to web services, desktop and mobile applications, and operating systems.
%For example, the ExP platform at Microsoft \cite{Kohavi2014,Kohavi2013,Kohavi2012TrustworthyExplained} supports experimentation in Bing, MSN, Cortana, Skype, Office, XBox, Windows, Edge Browser and more, running over ten thousand experiment treatments per year.
In a typical online controlled experiment that is evaluating a new feature, users are randomly assigned to the treatment group (exposed to the new feature) or the control group (not exposed to the new feature) as they come online to use the software product or service. The assignment remains consistent throughout the experiment. During the experiment period, we collect telemetry from all users and compute metric for both groups. We  conduct statistical tests to detect differences in metrics values between the treatment and control groups which are unlikely to be observed due to random chance. This establishes a causal relationship between the feature being tested and the measured changes in user behavior \cite{rubin2008objective,imbens2015causal}.  %A more detailed review of A/B testing can be found in \cite{Kohavi2013,Kohavi2014}. 

One key touchstone of trustworthiness of experimentation is external validity 
\cite{Campbell1957,campbell2015experimental,cook2002experimental,sabbaghi2018model} -- can results observed during an experiment period still hold when the new feature being tested is rolled out to the entire user population in the future? There can be multiple factors that affect external validity, such as the novelty effect and the weekday/weekend effect. While such factors are well recognized, there could be other neglected yet common effects that play an important role in determining external validity. In this paper, we highlight \emph{heavy-user bias}, which describes the phenomenon that frequent users are more likely to be included in an experiment than infrequent ones, rendering the estimated average treatment effect biased. To our best knowledge, it has not been formally discussed in the existing data mining literature, and we hope that this paper can raise the community's awareness on this important issue.

The remainder of the paper is organized as follows. Section \ref{sec:framework} discusses the concepts of external validity and heavy user bias in more depth and introduces necessary notations and assumptions. Section \ref{sec:methodology} derives the closed-form expression of the heavy-user bias, proposes a bias-adjusted estimator based on jackknife \cite{Quenouille56,Tukey1958,Rao1965AMethod,Miller74,Kunsch1989}, and illustrates the performance of the bias-adjusted estimator via simulated examples. Section \ref{sec:conclusion} concludes the paper by summarizing the progresses made in this on-going project, and discussing practical challenges and future research directions.

\section{Preliminaries}
\label{sec:framework}

\subsection{External validity and heavy user bias}
External validity, also known as generalizability \cite{Stuart11}, refers to the problem generalizing the findings from the sample units included in the experiment to a larger inference population. External validity is an important problem in causal inference and several papers studied the external validity under a variety of different scenarios such as politics \cite{Stuart11} and education \cite{Tipton14a}. 
%As shown , we use the bias of an A/B test to quantify external validity. A small bias means the A/B test is externally valid and a large bias means it is invalid.
%For instance, the paper \cite{Tipton14b} proposed a generalizability index that tries to quantify the similarity between the units in the experiment and the units outside the experiment. In those papers, the external validity issues arise because it is usually unrealistic to randomly decide if each unit should be included in the experiment. Therefore, a lot of efforts have been put into selecting a sample that is representative of the whole population. Nonetheless, for A/B testing, it is very straightforward to randomly decide which user should be included in the experiment. This is usually done via a "coin-flip" process: For each user that comes online, the system decides whether this user should be in treatment or control by internally calling a hashing function using the user's ID as input. Therefore, while the existing papers on external validity helps from a conceptual level, the proposed methods are not directly applicable.
In the context of A/B tests, the external validity of A/B tests could be affected by a variety of factors, such as novelty/primacy effects \cite{Sheinin11} or weekday/weekend effects.

Heavy user bias is another important yet often overlooked factor that affects the external validity of A/B tests. To illustrate what is heavy user bias, let us consider a simple example. Suppose a website has two hundred users, half are heavy users who use the website every day, and the other half are light users who use the website with 50\% probability each day.  However, if an online experiment is run for just one day, there would be around 150 users using the website.  The proportion of heavy users in the experiment sample will be 2/3. In other words, the proportion of heavy users for any A/B test is usually higher than that for the whole population. One simple way to mitigate heavy user bias would be to run the experiment long enough so that the proportion of heavy users and light users remain stable. However, under mild conditions, we show that the heavy user bias is usually inversely proportional to the length of the experiment $k$. This means when the experiment duration doubles, the heavy user bias is only reduced by half. Furthermore, long-term experiments are known to be prone to other critical issues \cite{pitfalls-long-term-online-experiments}. Therefore,  running experiments for longer periods might not be practical. 
%While those effects are well recognized and have some quick fixes, there could be other overlooked factors affecting external validity that would be crucial to take into account to make A/B testing results more trustworthy.
%
%\subsection{Continuous analysis}
% * <somit.gupta@gmail.com> 2018-08-09T21:18:45.668Z:
% 
% I am debating if we require to talk about the fixed duration analysis case at all in the paper. This just begs the question that is fixed duration analysis better in terms of bias? We are not answering that in the paper. 
% 
% I would vote to remove the reference to fixed duration analysis. 
% 
% ^.

%In A/B testing, a crucial question is what users should be included in the experiment. Users might appear at different days while the experiment duration is usually limited. In continuous analysis, we include every user who visits the product during the experiment period. If a user visited at the first day of the experiment, the analysis will include that user's data till the end of the experiment. In contrast, if a user did not appear until the last day of the experiment, the analysis will include the data of that user for the last day. The advantage of continuous analysis is ``no data left behind,'' where we used all the user data during the experiment period at hand. %The potential downside of that is the possible heterogeneous treatment effect because of that.

\subsection{Notations and assumptions}
 
\begin{table}[!h]
\centering
\begin{tabular}{ll}
\hline 
Notation & Explanation \\
\hline
\(\outcome(t)\) &  the observed outcome of user \(u\) at day \(t\).\\
\(\istreated{u}\) & whether the user \(u\) is in the treatment group.\\
\(\used(t)\) & whether the user \(u\) used the product at day \(t\).\\
\(t_u^0\) & the first time user $u$ shows up, i.e. $\min\{t \mid \used(t) = 1\}$\\
\(\effect(t)\) & the treatment effect for user \(u\) at day \(t\).\\
\(\control(t)\) & the control outcome for user \(u\) at day \(t\).\\
\(\whole\) & duration of the experiment (day 1 to day \(\whole\)).\\
\(\userN_\t,\userN_\c\)& number of users in the treatment/control group.\\
\hline
\end{tabular}

\end{table}
The table below describes the notation we will use for the rest of the paper.
To derive our theoretical result, we need to make the following assumptions. We will briefly comment why the introduced assumptions are reasonable in our application scenarios, and in Section \ref{sec:conclusion} we will discuss how to further relax the assumptions. 

\begin{assumption}
[stable unit treatment value assumption]
\label{assum:sutva}
One user's outcome is unaffected by other users' treatment assignments. In other words, different users do not interfere with each other.
\end{assumption}

\begin{assumption}[super population] For each user, its behavior can be characterized as a series of triplets $\{\used(t), \effect(t), \control(t)\}_{t=1}^\whole$. We assume that this series for each user is an i.i.d. sample from a super population with a probability distribution $\Psi$:
\begin{equation}
\label{eq:dgp}
P\{\used(t)=a_t, \effect(t)\leq b_t, \control(t)\leq c_t;t=1,\ldots,\whole\},
%= 
%\Psi((a_t, b_t, c_t)_{t=1}^\whole)
\end{equation}
where $a_t\in \{0, 1\}$, $b_t, c_t\in \mathbb R$ for $t\in \{1,\ldots, \whole\}$.
\end{assumption}

%The super population model is a commonly used model in causal inference literature \cite{Pearl08}. Its key feature is that it assumes that the potential outcome of each user is random. A natural consequence of the super population model is that the observed outcome of each user is independent. 
%Another commonly used model is the finite population potential outcomes framework \cite{Rubin1974EstimatingStudies,Neyman1923OnSummary.}, where the population size is finite and the potential outcomes $\{\effect(t), \control(t)\}_{t=1}^\whole$ are considered fixed. In that case, each user's outcome is not independent from one another. While the finite population model is more realistic especially when the population size is small, 
%The super population model gives a good approximation in large scale A/B testing scenarios where the number of users is usually at the scale of millions. The super population model also makes the mathematical analysis more readable and intuitive. A more detailed comparison between these two frameworks could be found in \cite{ding2017bridging}.

%Throughout the paper, we refer to user activity as the act of user using the product. We call a day $t$ as an active day for user $u$ if $R_u(t)=1$. Then we have the following assumption regarding $R_u(t)$ and $Z_u$:

\begin{assumption}[incremental experiment assumption]
\label{assum:user-activity-unchanged}
We assume that for each user $u$, the activity indicator $\used(t)$ is independent of the treatment assignment $\istreated u$.
\end{assumption}

\begin{remark}

Under Assumption \ref{assum:sutva}, the outcome of any user depends only on its own treatment assignment but not other users' treatment assignment. This assumption is reasonable when users do not interact with each other, e.g., users of search engines and operating systems. However, it could break for users that can interact and communicate, such as users of social media. The latter scenarios are beyond the scope of this paper.
%Although the heavy-user bias still exists, the current framework could not handle those scenarios. It is an interesting direction to study how the heavy-user bias could affect A/B tests when Assumption \ref{assum:sutva} breaks but that is not the focus of this paper.
%Without this assumption, if there are $N$ users, a user will have $2^N$ number of potential outcomes. Under Assumption \ref{assum:sutva}, the number of potential outcomes reduces to 2.

Assumption \ref{assum:user-activity-unchanged} implies that user's visit  \(\used(t)\) is not affected by whether a user is treated. In other words, we assume our experiment is incremental such that it does not change the frequency of users' visits.  This assumption could bring issues if a treatment significantly moves the number of days a user is active (i.e. uses the product). Therefore, 
it is important to test this assumption before analyzing the data using this framework. %Fortunately, this assumption is not hard to test. 
One plausible way would be to test whether the average active days per user is the same across treatment and control group. Based on our experience, most experiments do not affect the frequency of users' visits significantly \cite{Kohavi2014}.

Under Assumptions \ref{assum:sutva}--\ref{assum:user-activity-unchanged}, we can express the observed outcomes of the experimental units as
$$
\outcome(t) 
= \used(t) 
\{
\istreated{u} \effect(t) + \control(t)
\},
$$
which greatly facilitates the theoretical derivations going forward. 
\end{remark}

%There are several caveats about this model. First, although a unit is called "a user" in this paper, depending on the actual application scenario, it could be a user, a device, or a browser cookie. Similar ideas apply to $R_u(t)$. Although $R_u(t)$ is defined to be whether a user used the product at day $t$ in this paper, it could also be thought of as any other events, such as whether a user viewed a specific page.
%Second, when a user did not use the product, $\outcome(t)$ would be zero. That means there are two scenarios where $\outcome(t)$ can be zero. The first scenario is that a user did not use the product. The second scenario is that a user used the product but has outcome zero. We can tell the difference between these two scenarios by looking at $\used(t)$. Therefore, setting $\outcome(t) = 0$ when $\used(t) = 0$ will not cause any confusion. %If a user did not show up, $\used(t) = 0$; otherwise, it is one. That will become more clear when we define the specific metrics in the following section .

\section{Methodology}
\label{sec:methodology}

\subsection{Metric and point estimation}
At the end of an A/B test, we compute metrics to estimate the impact of the treatment on user behavior and make ship decisions. For example, click-through rate (CTR) is a common metric for search engines to measure the effectiveness of online advertisements. In this paper, we focus on the \emph{scaled single average}
$$
\frac{1}{k}
\E 
\left\{
\sum_{t=1}^\whole\outcome(t)\Big| \istreated u = z
\right\}
\quad
(z=1, 0)
$$
which are arguably the most common metric type in A/B testing\footnote{We will discuss other types of metrics in Section \ref{sec:conclusion}}. Note that the expectation is calculated with respect to the data generating mechanism in \eqref{eq:dgp}, which means it is the average over all the users of the product. For this metric, the average treatment effect (ATE) is the difference between the metric for the treatment group and that for the control group:
\begin{equation}
\label{Eq:estimand_delta}
\ATE_\scaled
= 
\E \left\{\frac{1}{\whole}\sum_{t=1}^\whole\used(t)\effect(t)\right\}.
\end{equation}
We can estimate $\ATE_\scaled$ by the corresponding difference-in-sample-means derived from the observed data:
\begin{equation}
\label{eq:naive-estmator}
\reallywidehat {\ATE}_\scaled 
= 
\frac{
\sum_{u:\istreated{u}=1} \sum_{t=1}^\whole\outcome(t)
}
{
\whole\userN_\t
} 
- 
\frac{
\sum_{u:\istreated{u}=0} \sum_{t=1}^\whole\outcome(t)
}
{
\whole\userN_\c
}.
\end{equation}
Note that, in this equation, we only include $\userN_\t + \userN_\c$ users that appear during the experiment.

\subsection{Heavy-user bias}
We define the \emph{heavy-user bias} of the estimator $\reallywidehat {\ATE}_\scaled$ estimating $\ATE_\scaled$ as the difference between the expected value of the estimator and the estimand: 
$
\E (\reallywidehat {\ATE}_{\scaled}) 
- 
\ATE_\scaled$. 
 
Because only users who appear between day 1--$\whole$ are included in the experiment, the expectation of the point estimate of $\reallywidehat {\ATE}_\scaled$ is:
\begin{equation}
\label{Eq:expected_ATE_c}
\E 
\left(
\reallywidehat {\ATE}_\scaled
\right)
= 
\E
\left\{
\frac{1}{\whole}\sum_{t=1}^\whole \effect(t)\used(t)\Big| t_u^0 \leq \whole
\right\},
\end{equation}

The \emph{heavy-user bias} is a (potentially complex) function of the data generating process in \eqref{eq:dgp}. To make the problem somewhat tractable and concrete, we propose a straightforward yet practical user behavior model, under which we derive the closed-form expressions of the heavy-user bias. 
We assume a fixed population, within which there exists user heterogeneity for both heavy and low activity frequencies and outcomes. 
%We acknowledge that this model does not capture the evolving dynamics of the user population for long-standing on-line applications and services, but for a relatively short time period, our model could be a reasonable approximation of user populations.

\begin{model}[Fixed population with user heterogeneity] 
\label{model:1}
%In many experiments, users' treatment effect and their activity can be correlated, e.g., heavy-user and light users react differently. 
We use the following model to reflect the heterogeneity on both user activity frequencies and outcomes:
\begin{itemize}
\item $\used(t)$ for a user $u$ on day $t$ is i.i.d. from a Bernoulli random variable with success probability $p\sim f(\cdot)$. 
\item The expectation of the treatment effect for user $u$ is $\E \effect(t) = \tau(p)$. It implies that the treatment effect could be different for users with different activity parameter $p$ but remains the same across all days.
\item Similarly, the expected control outcome of a user is $\E \control(t) = c(p)$.
\end{itemize}
\end{model} 
As demonstrated in the following lemma, Model \ref{model:1} allows us to derive the closed-form expressions for both $\ATE_\scaled$ and $\E(\reallywidehat {\ATE}_{\scaled}),$ and therefore rigorously quantify the heavy-user bias. 

\begin{lemma}
\label{lemma:1}
Under Assumptions \ref{assum:sutva}--\ref{assum:user-activity-unchanged} and Model \ref{model:1},
\begin{align}
\label{Eq:ATE_model_1}
\ATE_\scaled = &\int_{0}^1\tau(p) p f(p)dp,
\end{align}
and
\begin{align}
\label{Eq:3}
\E\left(\reallywidehat {\ATE}_{\scaled}\right) 
=& \dfrac{
\int_{0}^1\tau(p) p f(p)
dp
}
{
\int_{0}^1f(p)
\left\{
1 - (1 - p)^{\whole}
\right\}
dp
}.
\end{align}
\end{lemma}

\begin{proof}
First, \eqref{Eq:ATE_model_1} holds by the definition of Model \ref{model:1}. Second, based on \eqref{Eq:expected_ATE_c},
\begin{align}
\E
\left(
\reallywidehat {\ATE}_{\scaled} 
\right) 
= &\E\left\{\effect(1)\used(1)\Big| t_u^0 \leq \whole\right\}\\
= &\int_0^1\tau(p)\frac{p}{1 - (1-p)^\whole}f(p| t_u^0 \leq \whole)dp, \label{Eq:1}
\end{align}
where $f(p|t_u^0\leq \whole)$ is the density of the user activity probability $p$ conditioned on $t_u^0\leq \whole$. Using the Bayes' formula, we have
\begin{align}
f(p|t_u^0\leq \whole) = & \frac{f(p) P(t_u^0\leq \whole|p)}{\int_0^1f(p) P(t_u^0\leq \whole|p)dp}
= \frac{f(p) (1 - (1 - p)^\whole)}{\int_0^1f(p) (1 - (1 - p)^\whole)dp}. \label{Eq:2}
\end{align}
By \eqref{Eq:1} and \eqref{Eq:2}, we complete the proof.
\end{proof}

\begin{remark}
Intuitively, the heavy-user bias arises because light users are less likely to show up during the experiment and therefore are under-represented. If we run the experiment long enough ($\whole\to\infty$), then
\begin{equation*}
\lim_{\whole\to\infty}
\E
\left(
\reallywidehat {\ATE}_{\scaled}
\right) 
= 
\int_{0}^1\tau(p) p f(p)dp 
= 
\ATE_\scaled.
\end{equation*}
However, for a finite period $\whole$, a user with activity probability $p$ has probability 
$
1 - (1 - p)^\whole
$
to show up during the experiment. That shows the proportion of heavy-users during the experiment are higher than that in the whole population.
\end{remark}

With the help of Lemma \ref{lemma:1}, we can now present the main result of this paper.
\begin{prop}
\label{Prop:model1_single_average} Under Assumptions \ref{assum:sutva}--\ref{assum:user-activity-unchanged} and Model \ref{model:1}, if $f(\cdot), \tau(\cdot)$ has gradient and their gradients are continuous, then
\begin{equation*}
\E
\left(
\reallywidehat {\ATE}_{\scaled}
\right) 
- \ATE_\scaled
= 
\ATE_\scaled f(0)\cdot\whole^{-1} + O(\whole^{-2}).
\end{equation*}
\end{prop}

\begin{proof}
Denote $f^\prime$ to be $f$'s gradient. Because $f^\prime$ is continuous on the compact set $[0, 1],$ $f^\prime$ must be uniformly bounded on that set. That implies there exists a positive constant $C > 0$ such that
$
|f(p) - f(0)|\leq C\cdot p
$
for all $p \in [0,1].$ Therefore
\begin{equation*}
\left|\int_{0}^1f(p)(1 - p)^{\whole}dp - \int_{0}^1f(0)(1 - p)^{\whole}dp\right| 
\leq 
\int_{0}^1C\cdot p(1 - p)^{\whole}dp,
\end{equation*}
which implies that
\begin{align}
\left|\int_{0}^1f(p)(1 - p)^{\whole}dp - \frac{f(0)}{\whole+1}\right| \leq \frac{1}{(\whole+1)(\whole+2)} = O(\whole^{-2}).
\label{Eq:proof_sim1_bias_1}
\end{align}
By \eqref{Eq:3} in Lemma \ref{lemma:1} and \eqref{Eq:proof_sim1_bias_1},
$
\E\left\{\reallywidehat {\ATE}_{\scaled}\right\}  
= 
\ATE_\scaled + f(0) \ATE_\scaled \whole^{-1} + O(\whole^{-2}),
$
which completes the proof.
\end{proof}

\begin{remark}
When there is no extremely light users $f(0) = 0$, it can be seen from the proposition that the first order bias of continuous analysis would be zero. Note that if a user has probability zero of showing up, it will never appear in the experiment. $f(0)$ should be thought of as the limit $\lim_{q\rightarrow 0} P(p\leq q)/q,$ which approximately represents users with very light activity. Based on our experience, for many online websites, the proportion of extremely light users is quite significant. 
\end{remark}

\subsection{Bias-adjusted estimator}

Having derived the heavy-user bias in Proposition \ref{Prop:model1_single_average}, we next propose a bias-adjusted estimator to replace the difference-in-means estimator in \eqref{eq:naive-estmator}. Our proposal is inspired by jackknife in classical literature\cite{Tukey1958, Kunsch1989, Miller74}, which usually serves as a generic tool to correct first-order biases. 

For any fixed experiment duration $k$, let $a$ be the true value of  the outcome/metric of interest and $h(k)$ be an estimator of $a$. Assume the the \emph{heavy user bias} of the estimator can be approximated by a function of $\whole$:
$
h(k) - a = b / \whole + O(\whole^{-2}),
$
where $b$ is a parameter. The key insight here is that we can use two points $h(k-1)$ and $h(k)$ to get a better estimate of $a$ with almost no bias of order $O(\whole^{-1})$:
$$
\hat a = k \cdot h(k) - (k - 1) \cdot h(k-1).
$$
For the scaled single average metric, the natural choice for $h(k)$ is the un-adjusted estimator 
$
\reallywidehat {\ATE}_{\scaled}.
$ 
To obtain $h(k-1)$, we can similarly calculate 
$
\reallywidehat {\ATE}_{\scaled}
$
using the data of first $k - 1$ days. Although this estimate is unbiased, it does not use all the data at hand and can suffer from a large variance. To reduce its variance, we repeat the above procedure by excluding data from day $j=1, \ldots, \whole,$ and average the results. We summarize the above procedure in Algorithm \ref{Alg:jack_cont_period}.
\begin{algorithm}[htb]
\small
\caption{Bias-adjusted estimator}
\label{Alg:jack_cont_period}
\begin{algorithmic}[1]
\Require Data $= \{(\used(t), \outcome(t), \istreated{u})\}_{u}$ for days $j=1, \ldots, \whole$. 
\State $\reallywidehat {\ATE} \gets \mathrm{Calculate} \; \reallywidehat {\ATE}_{\scaled}$
\For{$j = 1,\ldots, \whole$}
\State Get new data-set by excluding data on day $j:$
\State $\mathrm{Calculate} \; \reallywidehat {\ATE}_{\scaled}$ on new data-set
\EndFor
\State $\bar \ATE \gets \frac{1}{\whole}\sum_{j} \reallywidehat \ATE_{(j)}$.
\State $\reallywidehat \ATE_{\textrm{jack}} \gets \whole \reallywidehat \ATE - (\whole - 1)\bar \ATE$.
\State \textbf{return} estimated mean $\reallywidehat \ATE_{\textrm{jack}}$ and its variance $\frac{k}{k - 1} \sum_{j=1}^\whole (\reallywidehat \ATE_{(j)} - \bar \ATE)^2$.
\end{algorithmic}
\end{algorithm}

\begin{prop}\label{Prop:jack_1}
Under Assumptions \ref{assum:sutva}--\ref{assum:user-activity-unchanged} and Model \ref{model:1}, the heavy-user bias of the bias-adjusted estimator in Algorithm \ref{Alg:jack_cont_period} is
$$
\E \reallywidehat \ATE_\jack - \ATE_\scaled 
= 
O(\whole^{-2}).
$$
\end{prop}

\begin{proof}
After excluding day $j,$ the remaining data can be viewed as from a $\whole - 1$ day experiment. Thus we apply Proposition \ref{Prop:model1_single_average} to obtain the expectation of  
$
\reallywidehat \ATE_{(j)},
$
the difference-in-means estimator on data excluding day $j:$ 
$
\E \reallywidehat \ATE_{(j)} = \ATE_{\scaled} + \frac{1}{\whole - 1}\ATE_{\scaled} f(0) + O(\whole^{-2}).
$
Therefore, the expectation of the bias-adjusted estimator is
\begin{equation*}
\E \reallywidehat \ATE_\jack 
= \whole\E \reallywidehat \ATE_\scaled - \frac{(\whole - 1)}{\whole}\sum_j \E \reallywidehat \ATE_{(j)} = \ATE_{\scaled} + O(\whole^{-2}).
\end{equation*}
\end{proof}

\subsection{Simulated examples}
\label{subsec:simu}
To demonstrate the advantages of the bias-adjusted estimator
$
\reallywidehat \ATE_\jack
$
in Proposition \ref{Prop:jack_1}, we present two simulated examples mimicking real-life A/B testing scenarios\footnote{We provide the source code of the simulations in \url{https://github.com/shifwang/On-Heavy-user-Bias-in-AB-Testing}}. For both examples, the experiment lasts for 14 days, the treatment and control groups each contains 1000 units, and each unit uses the product with probability $p$ for day $t=1, \ldots, \whole$, where $p$ is generated from  $\mathrm{Uniform}(0, 1)$. The difference is how we generate the outcomes. To be specific, if user $u$ uses the product on day $t,$ for Example 1 we let 
$$
\outcome(t) = 
\left\{
\begin{array}{ll}
1 + p + N(0, 0.01^2) & \text{if treated}\\
1 + N(0, 0.01^2) & \text{otherwise}
\end{array}
\right. ,
$$
and for Example 2 we let
$$
\outcome(t) = 
\left\{
\begin{array}{ll}
1 + (1 + \frac{1}{10\cdot U_u(t)})\cdot p + N(0, 0.01^2) & \text{if treated}\\
1 + N(0, 0.01^2) & \text{otherwise}
\end{array}
\right.,
$$
where $U_u(t)$ is the number of days the user $u$ used the product. For both examples, the ground truth 
$
\Delta_\scaled = 1/3.
$
While Example 1 is strictly under Model \ref{model:1}, Example 2 contains the novelty effect, which violates the assumptions in Model \ref{model:1}. By leveraging the two examples, we aim to examine both the accuracy and the robustness of the proposed bias-adjusted estimator. 

For both examples, we repeat the data generation process 100 times. For each simulated data-set, we compute the original difference-in-means estimator 
$
\reallywidehat {\ATE}_{\scaled},
$
the bias-adjusted estimator
$
\reallywidehat {\ATE}_\jack,
$
and the block bootstrap estimator\cite{Kunsch1989}. We report the biases and the standard deviations of the three estimators in Table \ref{Tab:sim_result}. The bias-adjusted estimator produces the smallest bias in both examples. Unfortunately, we do not have an answer why jackknife adjustment works better than bootstrap under our simulated settings.
\begin{table}[tb]
\centering
\footnotesize
\caption{Average biases and their standard errors}
\label{Tab:sim_result}
\vspace{-2mm}
\begin{tabular}{@{}lcc@{}}
\toprule
  & Example 1          & Example 2     \\ 
    \hline
Bias of original estimator  &  0.0220 (0.0023)  & 0.0373 (0.0020) \\
Bias of bias-adjusted estimator &  -0.0022 (0.0026)  & 0.0132 (0.0023)    \\
Bias of block-bootstrap estimator & 0.0080 (0.0025)  & 0.0232 (0.0022)  \\ \bottomrule
\end{tabular}
\end{table}

\section{Concluding Remarks}
\label{sec:conclusion}

In this paper, we highlighted that the heavy-user bias could affect external validity significantly, and would like to raise the awareness of the data mining community on this issue. To be more specific, we demonstrated that the heavy-user bias exists in A/B testing due to the limited length of an experiment, and proposed a bias-adjusted estimator based on jackknife. Under the fixed population user heterogeneity model, we showed that jackknife estimators could correct the first order heavy-user bias. We conducted simulation studies to illustrate the advantages of the proposed methodology.

We summarize two lines of active research we have been conducting on this on-going project. First, besides \eqref{Eq:estimand_delta} there are other types of metrics. For example, double average:
$
\ATE_\double 
=\E
\left\{
\frac{\sum_{t=1}^\whole\used(t)\effect(t)}{\sum_{t=1}^\whole\used(t)}
\right\}.
$
Fortunately, we can leverage the same techniques to derive the heavy-user bias. Second, we applied the bias-adjusted estimator to several empirical data-sets and we found that, comparing with short term A/B testing results, our new estimator is closer to the long term A/B testing results. %We are verifying and organizing the new results, and will communicate them in a separate report. 

To conclude the paper, we outline several future research directions to achieve the holy grail of ensuring external validity. First, it is important to understand the joint effect of multiple factors that affect external validity. The simulated examples showed that our proposed estimator provided a more accurate estimate in the presence of both the heavy user bias and the novelty effect. However, we might need to consider other possible effects in practice. Second, we can generalize the current methodology to study fairness in A/B testing, an important topic in machine learning and artificial intelligence \cite{Corbett2017}. Third, it would be interesting to extend the current study to the network setting with user interference. 

\section{Acknowledgement}
The major work is done during Yu's summer internship at ExP team at Microsoft Inc. in 2018. Partial supports are gratefully acknowledged from ARO grant
W911NF1710005, ONR grant N00014-16-1-2664, NSF grants
DMS-1613002, IIS 1741340, and CSoI under grant agreement CCF-0939370.

\bibliographystyle{ACM-Reference-Format}
\nocite{}
%\small
\tiny
\bibliography{lib_jiannl.bib,references3.bib}

\end{document}